\documentclass[twoside,twocolumn,english,aps,pra,reprint,superscriptaddress]{revtex4-2}
\usepackage[T1]{fontenc}
\usepackage[latin9]{inputenc}
\usepackage{color}
\usepackage{babel}
\usepackage{amsmath}
\usepackage{amsthm}
\usepackage{graphicx}
\usepackage[unicode=true,pdfusetitle,
 bookmarks=true,bookmarksnumbered=false,bookmarksopen=false,
 breaklinks=false,pdfborder={0 0 1},backref=false,colorlinks=false]
 {hyperref}

\makeatletter

\theoremstyle{plain}
\newtheorem{fact}{\protect\factname}

\renewcommand\[{\begin{equation}}
\renewcommand\]{\end{equation}}
\usepackage{caption}

\usepackage{times}

\makeatother

\providecommand{\factname}{Fact}

\begin{document}
\global\long\def\ket#1{|#1\rangle}%
\global\long\def\bra#1{\langle#1|}%
\global\long\def\braket#1#2{\left\langle #1\middle|#2\right\rangle }%
\global\long\def\diff{\mathrm{d}}%
\global\long\def\inf{\infty}%
\global\long\def\pd#1#2{\frac{\partial#1}{\partial#2}}%
\global\long\def\Tr{\mathrm{Tr}}%
\title{Hardware-efficient variational quantum algorithm in trapped-ion quantum computer} 
\author{J.-Z. Zhuang} 
\affiliation{Center for Quantum Information, Institute for Interdisciplinary Information Sciences, Tsinghua University, Beijing 100084, PR China} 
\affiliation{Shanghai Qi Zhi Institute, AI Tower, Xuhui District, Shanghai 200232, China}
\author{Y.-K. Wu} 
\email{wyukai@mail.tsinghua.edu.cn} 
\affiliation{Center for Quantum Information, Institute for Interdisciplinary Information Sciences, Tsinghua University, Beijing 100084, PR China} 
\affiliation{Shanghai Qi Zhi Institute, AI Tower, Xuhui District, Shanghai 200232, China}
\affiliation{Hefei National Laboratory, Hefei 230088, PR China} 

\author{L.-M. Duan} 
\email{lmduan@tsinghua.edu.cn} 
\affiliation{Center for Quantum Information, Institute for Interdisciplinary Information Sciences, Tsinghua University, Beijing 100084, PR China} 
\affiliation{Hefei National Laboratory, Hefei 230088, PR China} 
\affiliation{New Cornerstone Science Laboratory, Beijing 100084, PR China}

\captionsetup[figure]{labelfont={bf},name={FIG.},labelsep=period, justification=raggedright} 
\begin{abstract}
We study a hardware-efficient variational quantum algorithm ansatz
tailored for the trapped-ion quantum simulator, HEA-TI. We leverage
programmable single-qubit rotations and global spin-spin interactions
among all ions, reducing the dependence on resource-intensive two-qubit
gates in conventional gate-based methods. We apply HEA-TI to state
engineering of cluster states and analyze the scaling of required
quantum resources. We also apply HEA-TI to solve the ground state
problem of chemical molecules $\mathrm{H_{2}}$, $\mathrm{LiH}$ and
$\mathrm{F_{2}}$. We numerically analyze the quantum computing resources
required to achieve chemical accuracy and examine the performance
under realistic experimental noise and statistical fluctuation. The
efficiency of this ansatz is shown to be comparable to other commonly
used variational ansatzes like UCCSD, with the advantage of substantially
easier implementation in the trapped-ion quantum simulator. This approach
showcases the hardware-efficient ansatz as a powerful tool for the
application of the near-term quantum computer.
\end{abstract}
\maketitle

\section{Introduction}

Trapped-ion systems have emerged as one of the leading platforms for
quantum information processing \citep{bruzewiczTrappedionQuantumComputing2019,monroeProgrammableQuantumSimulations2021}.
Apart from long coherence time and high-fidelity initialization, single-qubit
operation, and readout \citep{leibfriedQuantumDynamicsSingle2003},
trapped ions support long-range spin-spin interactions mediated by
collective phonon modes under the Coulomb force between them \citep{ciracQuantumComputationsCold1995,sorensenEntanglementQuantumComputation2000}.
Recent experiments have demonstrated tunable coupling range and patterns
of long-range interactions across large-scale ion crystal systems
\citep{monroeProgrammableQuantumSimulations2021,guoSiteresolvedTwodimensionalQuantum2024}.
Such interaction distinguishes trapped-ion systems from other quantum
computing platforms \citep{saffmanQuantumComputingAtomic2016,huangSuperconductingQuantumComputing2020,kjaergaardSuperconductingQubitsCurrent2020,wuConciseReviewRydberg2021}
in which qubits mainly possess short-range nearest-neighbor interactions,
and it enables the study of classically intractable many-body quantum
dynamics like prethermalization \citep{neyenhuisObservationPrethermalizationLongrange2017},
information scrambling \citep{landsmanVerifiedQuantumInformation2019,joshiQuantumInformationScrambling2020},
and dynamical phase transitions \citep{jurcevicDirectObservationDynamical2017,zhangObservationManybodyDynamical2017}.

Recent advances in trapped-ion systems have also enabled the use of
variational quantum algorithms (VQAs) \citep{mollQuantumOptimizationUsing2018,tillyVariationalQuantumEigensolver2022}
to find the ground states of non-trivial Hamiltonians \citep{kokailSelfverifyingVariationalQuantum2019,zhuTrainingQuantumCircuits2019,zhaoOrbitaloptimizedPaircorrelatedElectron2023,ollitraultEstimationElectrostaticInteraction2024}.
The ground state problem is important across various disciplines \citep{jensenIntroductionComputationalChemistry2017,kohnNobelLectureElectronic1999,ccernyQuantumComputersIntractable1993},
with various quantum algorithms proposed to address it \citep{yarkoniQuantumAnnealingIndustry2022,tillyVariationalQuantumEigensolver2022,verstraeteQuantumComputationQuantumstate2009}.
Among them, the VQA is specifically designed to reduce the requirements
on quantum devices by integrating quantum computing with classical
optimization \citep{yungTransistorTrappedionComputers2014,peruzzoVariationalEigenvalueSolver2014}.
As shown in Fig. \ref{fig:model}, the quantum system runs a sequence
of parametrized operations to generate variational trial states. The
parameters are controlled and optimized iteratively by a classical
computer to minimize a cost function according to the quantum system's
measurement outcome. This method has been applied to a wide range
of key problems, including those in quantum chemistry \citep{peruzzoVariationalEigenvalueSolver2014,googleaiquantumandcollaboratorsHartreeFockSuperconductingQubit2020},
condensed matter physics \citep{liuProbingManybodyLocalization2023},
and quantum field theory \citep{kokailSelfverifyingVariationalQuantum2019}.
\textcolor{blue}{}

\begin{figure}
\includegraphics[width=1\columnwidth]{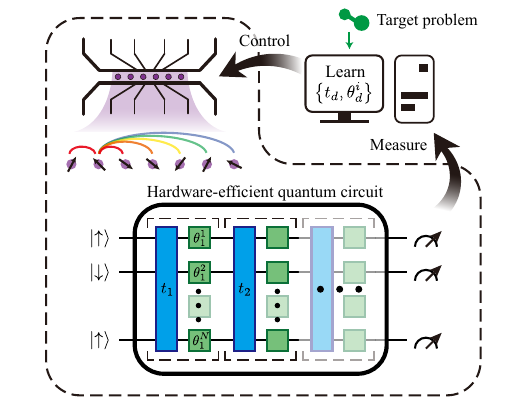}\caption{Scheme for the variational quantum algorithm in a trapped-ion quantum
computer: device, quantum circuit, and measurement feedback. Bottom:
Variational control parameters are passed to the trapped-ion quantum
simulator, which generates trial states through quantum dynamics from
a set of resource Hamiltonians, consisting of global entangling operations
(blue boxes) and single-qubit operations (green boxes). Energy and
gradient are obtained by quantum projective measurements on the Pauli
basis. A classical computer stores the obtained measurement results
and allows for data re-evaluation for different Hamiltonian parameters.
It also optimizes the parameters $\left\{ t_{d},\theta_{d}^{i}\right\} $
based on the measurement outcomes. \label{fig:model}}
\end{figure}

Despite the reduced requirement on fault-tolerant universal quantum
computing, VQA remains challenging to implement across all current
quantum platforms. Numerous efforts have focused on designing variational
ansatz to improve representability and trainability \citep{mcardleQuantumComputationalChemistry2020,cadeStrategiesSolvingFermiHubbard2020,wiersemaExploringEntanglementOptimization2020}.
However, these approaches often require an extensive array of quantum
logical gates to compose complex operations tailored to specific problems
\citep{romeroStrategiesQuantumComputing2018,zhuGenerativeQuantumLearning2022,backesDynamicalMeanfieldTheory2023}.
Indeed, while recent experiments are already demonstrating a remarkable
level of quantum control \citep{aruteQuantumSupremacyUsing2019,wuStrongQuantumComputational2021,kimEvidenceUtilityQuantum2023,bluvsteinLogicalQuantumProcessor2024},
the realization of high-fidelity multi-qubit gates in large-scale
systems remains a long-term challenge for all quantum computing platforms.
Therefore, suitable ansatz is highly needed to exploit the strengths
and avoid the weaknesses of a device. Here, we investigate a hardware-efficient
ansatz for trapped-ion systems that addresses these challenges by
leveraging the inherent long-range interactions among ions.

The hardware-efficient ansatz (HEA) for variational circuits is composed
of alternating layers of native gates tailored to a specific quantum
computer \citep{mollQuantumOptimizationUsing2018,tillyVariationalQuantumEigensolver2022,kandalaHardwareefficientVariationalQuantum2017}.
Originally, HEA was designed for superconducting circuits \citep{kandalaHardwareefficientVariationalQuantum2017}.
It utilizes single-qubit gates and entanglers, with the entanglers
constructed by ladders of nearest-neighbor two-qubit entangling gates.
In trapped-ion systems, a natural choice of the entangler layer is
the global evolution under spin-spin interaction Hamiltonian \citep{monroeProgrammableQuantumSimulations2021}.
Such ansatz was experimentally implemented in previous work \citep{kokailSelfverifyingVariationalQuantum2019,paganoQuantumApproximateOptimization2020},
but detailed theoretical analyses like the scaling of required resources
and the sensitivity to the experimental errors have not been investigated. 

In this work, we investigate a hardware-efficient variational ansatz
tailored for the trapped-ion quantum simulator (HEA-TI). In Sec. \ref{sec:Hardware-efficient-variational-a},
we describe our model for trapped-ion devices and then introduce the
HEA-TI. We also prove the universality of this ansatz. In Sec. \ref{sec:Preparation-of-Cluster},
we compare the HEA-TI with the conventional gate-based ansatz. Specifically,
we apply HEA-TI to state engineering and study the required quantum
resources in different system sizes. In Sec. \ref{sec:Solving-Electronic-Structure-Pro},
we apply it to another problem of solving the ground state energy
of diatomic molecules, and we further analyze the sensitivity to errors
in Sec. \ref{sec:Error-Analysis}. We describe a generalization of
the parameter shift method to evaluate the gradient for global evolution
time in Appendix \ref{sec:Parameter-Shift-Method}. The details of
the chemistry problem encoding and the sampling error analysis for
experiment simulation are provided in Appendices \ref{sec:Problem-Encoding-and}
and \ref{sec:Sampling-Error-of}, respectively.

\section{Hardware-efficient variational ansatz for trapped ions \label{sec:Hardware-efficient-variational-a}}

\subsection{Model of trapped-ion quantum computer}

Consider a trapped-ion quantum simulator consisting of $N$ ions that
are confined as a one-dimensional string. Each ion encodes a qubit
in its internal electronic levels. The spin-spin coupling can be realized
by a bichromatic laser beam that off-resonantly couples all the ions
to the transverse phonon modes \citep{monroeProgrammableQuantumSimulations2021}.
Such interaction can be described by the transverse field Ising model
(TFIM)
\begin{equation}
H_{\text{TFIM}}=\sum_{i<j}J_{ij}\sigma_{x}^{i}\sigma_{x}^{j}+\sum_{i}B\sigma_{z}^{i},\label{eq:H_TFIM}
\end{equation}
 where $\sigma_{x(z)}^{i}$ are the Pauli operators on the $i$-th
ion, $B$ is an effective, uniform magnetic field produced by suitable
detuning of the bichromatic laser, $J_{ij}$ is the effective spin-spin
coupling, and the summation $\sum_{i<j}$ runs over all ion pairs.
The Ising coupling coefficients $J_{ij}$ depend on details of the
experimental setups, such as laser detuning, Rabi frequency on each
ion, phonon modes' frequency, and Lamb-Dicke parameter. 

Our approach does not rely on the accurate implementation of any specific
entangler and can be used under any interaction pattern that can generate
sufficient entanglement. Typically, $J_{ij}$ can be approximated
by $J_{ij}\approx J_{0}/|i-j|^{\alpha}$ with $\alpha\in[0.5,1.8]$
\citep{monroeProgrammableQuantumSimulations2021}. Although varying
$J_{ij}$ provides extra degrees of freedom that enhance the ansatz's
expressive power, their optimization is hard due to the lack of a
straightforward way to evaluate the gradient. Here, we fix them as
constant $J_{0}=1,\alpha=1.5$ for our device. 

Single-qubit operations can be conveniently implemented by individual
addressing laser on each ion. In practice, these operations are significantly
faster than global evolution. Therefore, it is possible to keep the
global evolution always on while inserting single-qubit operations
at desired times to implement our circuit, which further simplifies
the experimental control sequences and helps suppress the errors from
turning on and off the global interaction.

\subsection{Hardware-efficient ansatz for trapped-ion quantum computer \label{subsec:Hardware-efficient-ansatz-for}}

Given such a trapped-ion device, we introduce the HEA-TI. Consider
a target observable $O$, we aim to optimize the parametrized state
$\ket{\psi(\boldsymbol{\theta})}=U(\boldsymbol{\theta})\ket{\psi_{0}}$
for minimal expectation value $\bra{\psi(\boldsymbol{\theta})}O\ket{\psi(\boldsymbol{\theta})}$.
The initial state $\ket{\psi_{0}}$ is a product state on the computational
basis. $U(\boldsymbol{\theta})$ is implemented by $D$ layers of
quantum circuits, each constructed by single-qubit rotations on all
the qubits and a global evolution, as shown in Fig. \ref{fig:model}.
The global evolution entangles all the qubits in the system, and the
single-qubit rotations provide extra degrees of freedom. 

For general problems with no specific symmetry, we adopt $H_{\text{TFIM}}$
with $B=1$ for global evolution and arbitrary $SU(2)$ rotation $R_{x}(\theta^{1})R_{y}(\theta^{2})R_{x}(\theta^{3})$
for each single-qubit rotation. The global Hamiltonian evolution $e^{-iH_{\text{TFIM}}t_{d}}$
at the $d$-th layer is parametrized by the evolution time $t_{d}$
under fixed $H_{\text{TFIM}}$. The single-qubit rotations at the
$d$-th layer are parametrized by $\left\{ \theta_{d}^{i,1},\theta_{d}^{i,2},\theta_{d}^{i,3}\right\} $
where $i$ labels the qubit index. In total, there are $D(3N+1)$
independent parameters. 

For problems that exhibit symmetry, we can further reduce the computational
resource by restricting the optimization of parameters to the symmetry
sector of interest. The electronic-structure problem we will discuss
below exhibits charge conservation symmetry $\sigma_{\text{tot}}^{z}=\sum_{i}\sigma_{i}^{z}$.
We limit the single-qubit operation to be Z-axis rotation with rotation
angle $\theta_{d}^{i}$ and adopt global Hamiltonian evolution that
preserves the charge conservation symmetry. This can be achieved by
setting the magnetic field in Eq. \ref{eq:H_TFIM} to be large enough
that overwhelms the Ising interactions $B\gg J_{ij}$. Then, after
rotating-wave approximation and moving into a rotating frame, the
effective Hamiltonian of the global evolution is

\[
H_{\text{XY}}=\sum_{i<j}J_{ij}\left(\sigma_{+}^{i}\sigma_{-}^{j}+\sigma_{-}^{i}\sigma_{+}^{j}\right).
\]

All parameters are initialized randomly at the beginning of optimization.
They are then optimized iteratively by a classical computer based
on the measurement outcomes of the system. We employ gradient-based
methods for precise optimization. The gradient for single-qubit rotations
can be obtained using the parameter-shift method \citep{schuldEvaluatingAnalyticGradients2019}.
An extension of this method enables the calculation of the gradient
of $t_{d}$, assuming the addition of a single multi-qubit gate in
the circuit is allowed (for details, see Appendix \ref{sec:Parameter-Shift-Method}).
Thus, we allow gradient descent of $t_{d}$ when benchmarking the
expressive power of HEA-TI in Sec. \ref{sec:Preparation-of-Cluster}.
In Sec. \ref{sec:Solving-Electronic-Structure-Pro} and \ref{sec:Error-Analysis},
we fix $\left\{ t_{d}\right\} $ for easier experimental implementation.
However, it is important to note that $t_{d}$ can still be optimized
using non-gradient methods. 

\subsection{Universality of HEA-TI \label{subsec:Universality-of-HEA-TI}}

The HEA-TI constructed above possesses universal computing ability.
For theoretical proof, it suffices to show that any two-qubit entangling
gate can be decomposed into global evolutions under the TFIM Hamiltonian
and single-qubit rotations. Specifically, we adopt $H_{\text{TFIM}}$
with $B=0$ for global evolution. By inserting single-qubit rotations
during the global evolution, we can construct a spin echo and cancel
the unwanted phases for all other ion pairs, thereby obtaining any
desired two-qubit entangling gate \citep{bholeRescalingInteractionsQuantum2020}. 

Thus, arbitrary quantum states can be prepared using a sufficiently
large number of layers $D$. However, this construction is often quite
inefficient, as a significant amount of entanglement is canceled during
the circuit. As will be discussed in Section \ref{sec:Preparation-of-Cluster},
it is sufficient to prepare the desired quantum state with much shallower
circuits that make full use of the generated entanglement.

\section{Cluster State Preparation using HEA-TI \label{sec:Preparation-of-Cluster}}

To compare the feasibility and efficiency of the HEA-TI with the conventional
gate-based ansatz, we first study their performance in the quantum
state engineering problem. 

Our construction of HEA-TI immediately makes the preparation of a
specific set of quantum states convenient. For example, the state
$e^{-iH_{\text{TFIM}}t}\ket{\psi_{0}}$ can be prepared with one global
evolution layer in HEA-TI, whereas traditional gate-based schemes
would require $\mathcal{O}\left(N^{2}\right)$ two-qubit gates to
match the degree of freedom in $H_{\text{TFIM}}$. However, the efficiency
of HEA-TI for a wider range of states is in question. Here, we show
that HEA-TI is also able to efficiently generate multipartite entangled
states that are convenient to prepare by traditional gate-based schemes.

We consider the cluster state \citep{nielsenClusterstateQuantumComputation2006},
which can be prepared by applying a few layers of two-qubit gates
on a product state. In two-dimensional systems, cluster states are
the resource state of measurement-based quantum computing. Here, we
consider only one-dimensional (1D) cluster states $\ket{C_{N}}$ with
$N$ qubits for convenience in simulation. The target observable can
be set as 
\[
O_{\text{cluster}}=-X_{1}Z_{2}-\sum_{i=2}^{N-1}Z_{i-1}X_{i}Z_{i+1}-Z_{N-1}X_{N},
\]
in which each term is a stabilizer of the 1D cluster state. Therefore,
the unique ground state of this observable corresponds to the 1D cluster
state.

We adopt the ansatz with $H_{\text{TFIM}}$ and arbitrary single-qubit
rotations for general problems, as described in Sec. \ref{subsec:Hardware-efficient-ansatz-for}.
We further allow optimization of evolution time $t_{d}$ to better
understand the role of global evolution and how they can best enhance
precision. The variational parameters are initialized randomly and
optimized for the minimal expectation value of $O_{\text{cluster}}$,
and the optimization is repeated $20$ times to escape from local
minima. The final-state fidelity $F=\left|\braket{C_{N}}{\psi(\theta_{\text{opt}})}\right|^{2}$
for different circuit depths is shown in Fig. \ref{fig:state-engineering}.
For all system sizes $N$, we observe an exponential improvement in
accuracy as we increase the circuit depth. As $N$ or $D$ increases,
the improvement slows down, which is expected due to the more complex
search space. Nevertheless, for fixed $N=8$, after the initial slowdown
at small $D$, we confirm that the accuracy keeps exponential improvement.
Therefore, we conclude that HEA-TI efficiently prepares cluster states
with high accuracy.

\begin{figure}
\includegraphics[width=0.9\columnwidth]{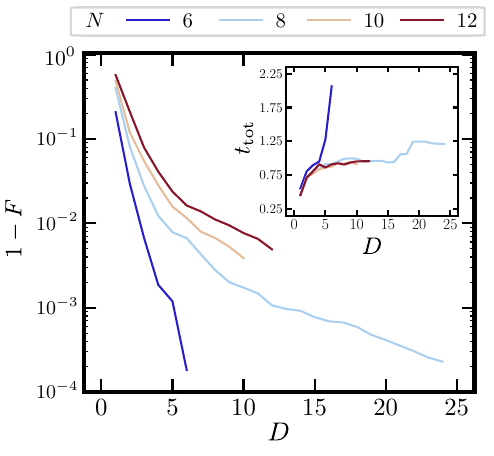}

\caption{Fidelity $F$ of 1D cluster state prepared by $D$-layer HEA-TI after
optimization. Due to the computational expense, we terminate the calculation
at an energy error of $10^{-3}$ for $N\protect\leq8$, and at a circuit
depth $D=N$ for $N>8$. For all system sizes $N$, we observe a consistent
exponential decrease in infidelity in the late part, indicating an
exponential improvement in accuracy as the circuit depth increases.
The inset shows the total global evolution time $\sum t_{d}$ as a
function of $D$, which exhibits saturating behavior for increasing
$D$ (except the smallest $N=6$ case). \label{fig:state-engineering}}
\end{figure}

We can understand the effectiveness of HEA-TI from the perspective
of entanglement. Unlike the construction for the theoretical proof
in Sec. \ref{subsec:Universality-of-HEA-TI}, the global evolution
is utilized efficiently rather than being canceled. We consider the
total evolution time $t_{\text{tot}}=\sum t_{d}$ of global interaction.
As shown in the inset of Fig. \ref{fig:state-engineering}, $t_{\text{tot}}$
in the variationally optimized circuit saturates as the circuit depth
increases, except for the smallest $N=6$ case, where $t_{\text{tot}}$
increases linearly as the accuracy improves exponentially. The saturation
behavior aligns with the constant entanglement in 1D cluster states.
This indicates that the optimal global evolution time is determined
by the entanglement in the target state, rather than increases with
the mere addition of circuit layers. The efficient use of entanglement
not only saves quantum resources but also brings a constraint on the
search space within a manageable scope of entanglement \citep{ortizmarreroEntanglementInducedBarrenPlateaus2021}.

Although we studied 1D cluster state that matches the topology of
our concerned 1D ion trap system, we expect no qualitative difference
when extending to 2D cluster states due to the long-range interaction
in global evolution, especially when a 2D ion crystal can be used
\citep{guoSiteresolvedTwodimensionalQuantum2024}. 

\section{Solving Electronic-Structure Problem using HEA-TI \label{sec:Solving-Electronic-Structure-Pro}}

Next, we study the performance of HEA-TI in industry-relevant problems
which have certain symmetry but a more complex underlying topology.
The results are then compared with the gate-based ansatz in a state-of-the-art
experiment. 

The electronic structure problem in quantum chemistry is one of the
most promising research areas for realizing practical quantum advantages
in the near future \citep{bhartiNoisyIntermediatescaleQuantum2022,caoQuantumChemistryAge2019,mcardleQuantumComputationalChemistry2020,alexeevQuantumcentricSupercomputingMaterials2023}.
As an important step in computing the energetic properties of molecules
and materials, the ground state problem aims to compute the ground
state energy of electrons in a given molecular to chemical accuracy,
i.e., $1.5$ milli-Hartree. Such problems can be formulated as the
interacting fermionic problem, where finding exact numerical solutions
by classical methods has a computational cost that scales exponentially
with the size of the system. Even few-atom electronic-structure problems
are interesting for quantum computers, making it attractive for applications
on current devices with limited ability \citep{hempelQuantumChemistryCalculations2018,zhaoOrbitaloptimizedPaircorrelatedElectron2023,guoExperimentalQuantumComputational2024,ollitraultEstimationElectrostaticInteraction2024}.

Here, we consider the ground state problem of three commonly used
molecules $\mathrm{H_{2}}$, $\mathrm{LiH}$, and $\mathrm{F_{2}}$
\citep{guoExperimentalQuantumComputational2024,zhaoOrbitaloptimizedPaircorrelatedElectron2023}.
The Hamiltonian of each electronic system is encoded into a target
observable, requiring $4,6$, and $12$ qubits, respectively. Details
of the encoding scheme are summarized in Appendix \ref{sec:Problem-Encoding-and}.

\begin{figure}
\includegraphics[width=0.85\columnwidth]{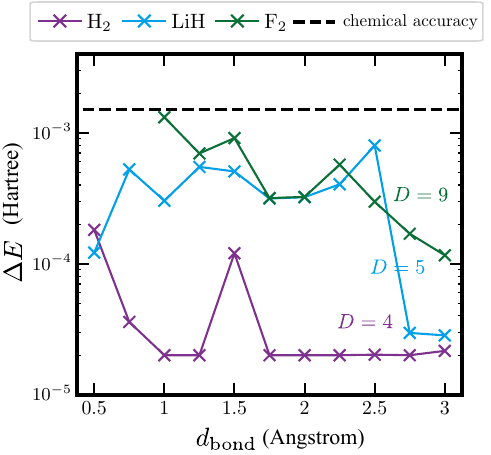}

\caption{Error of ground state energy $\Delta E$ compared to the exact results
for different molecules and bond distances $d_{\text{bond}}$. We
consider $\mathrm{H_{2}}$ (4 qubits, $D=4$), $\mathrm{LiH}$ (6
qubits, $D=5$), and $\mathrm{F_{2}}$ (12 qubits, $D=9$). For each
molecule, we show the least circuit depth to achieve chemical accuracy
(dashed line) for all bond distances. We fix the global evolution
time $t_{d}=0.4$ in all layers, thereby eliminating the need for
evaluating gradients of these parameters. The optimization results
are obtained from $20$ trials from random initial parameters. \label{fig:without-noise}
}
\end{figure}

We adopt the HEA-TI circuit structure for problems with $\sigma_{\text{tot}}^{z}$
symmetry, as described in Sec. \ref{subsec:Hardware-efficient-ansatz-for}.
To minimize the requirements on quantum hardware, for the remainder
of this paper, we do not evaluate the gradient of global evolution
time $t_{d}$ using multi-qubit gates. For simplicity and experimental
convenience, we fix the evolution time of all layers to a constant
$t_{d}\equiv t_{0}=0.4$, although $\left\{ t_{d}\right\} $ may still
be optimized using gradient-free methods.

After optimization, we show the error of obtained energy $\Delta E$
from the exact result. As shown in Fig. \ref{fig:without-noise},
with $D=4,5$, and $9$ layers of hardware-efficient operations, chemical
accuracy can be reached for all bond distances for the three molecules
$\mathrm{H_{2}}$, $\mathrm{LiH}$, and $\mathrm{F_{2}}$, respectively. 

To compare with the gate-based ansatz, we consider the UCCSD ansatz
\citep{romeroStrategiesQuantumComputing2018,anandQuantumComputingView2022},
another commonly used variational ansatz for electronic-structure
problems. By tailoring the circuit for each molecule individually,
one needs $10,22$, and $50$ two-qubit gates, or $9,20$, and $32$
layers of parallel gates, to achieve chemical accuracy \citep{guoExperimentalQuantumComputational2024}.
In this sense, we conclude that HEA-TI solves the chemistry problem
more efficiently than the gate-based UCCSD scheme. 

\section{Error Analysis \label{sec:Error-Analysis}}

As discussed above, HEA-TI is capable of solving the quantum chemistry
problem efficiently. However, under near-term quantum device capability,
the performance of HEA-TI is limited not only by the protocol's expressive
power but also by noise from various sources. In this section, we
evaluate our protocol's performance on near-term devices. Specifically,
we study two dominant error sources: sampling error of gradient and
drift of circuit parameters, and study how they influence the energy
outcome.

We select $D=4$ for $\mathrm{H_{2}}$ and $D=5$ for $\mathrm{LiH}$
to achieve chemical accuracy. For $\mathrm{F_{2}}$, we select $D=3$
to balance the expressive power of our protocol and the effect of
noise when implementing in experiments. 

\begin{figure*}
\includegraphics[width=1.8\columnwidth]{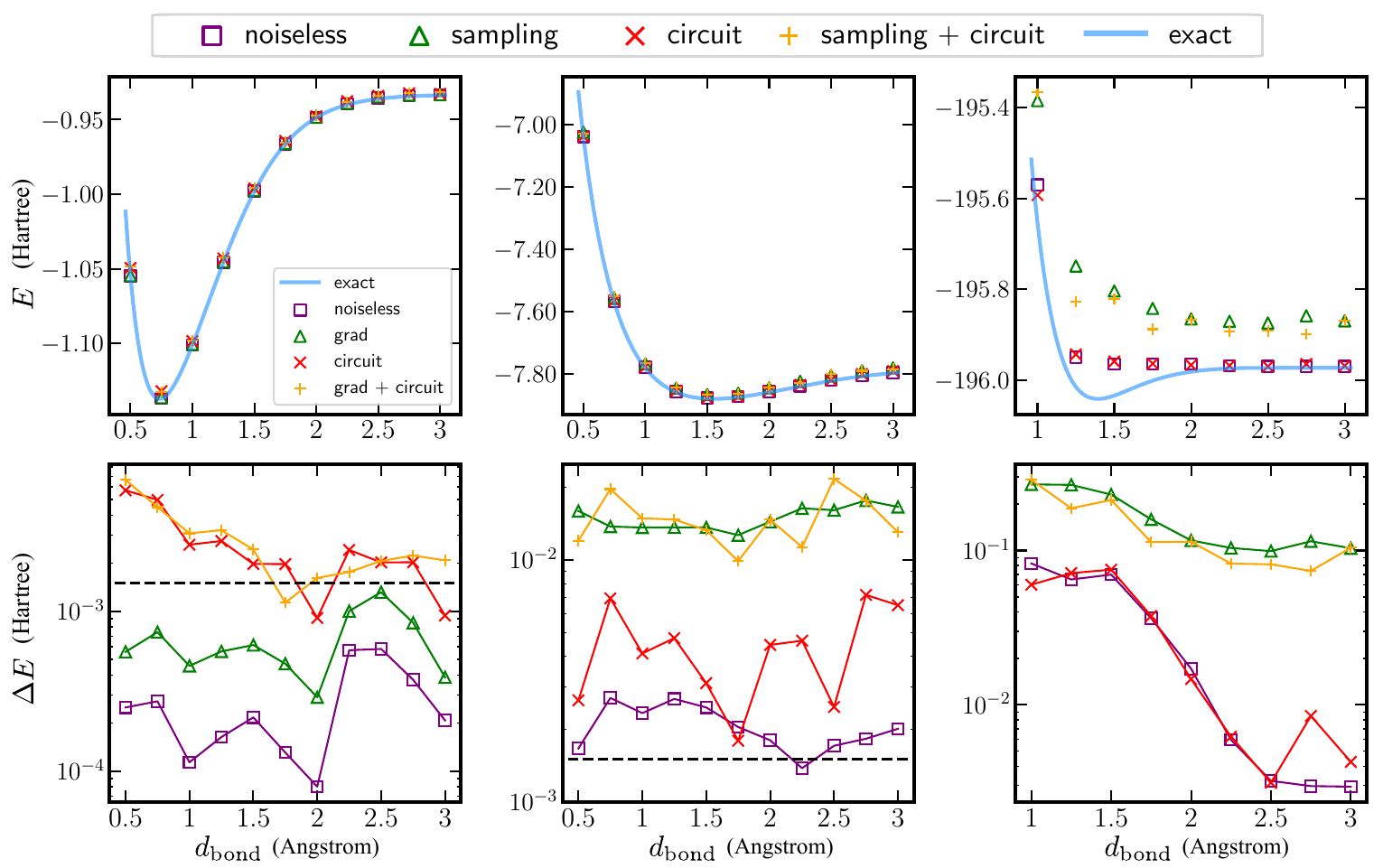}\caption{Ground state energy curves from simulated experimental implementation
for different molecules. (a-c) Average ground state energy curves
as functions of the bond distance for $\mathrm{H_{2}}$ (4 qubits,
$D=4$), $\mathrm{LiH}$ (6 qubits, $D=5$), and $\mathrm{F_{2}}$
(12 qubits, $D=3$). The results are averaged over $20$ trials from
random initial parameters. (d-f) Energy error compared to the exact
results. We compare the noiseless case (purple squares), only sampling
error in gradient (green triangles), only circuit drift error (red
crosses), and both errors (orange pluses). The exact ground state
energy (blue line) and the energy with chemical accuracy (dashed line)
are calculated theoretically as a reference. \label{fig:with-noise}
}
\end{figure*}

\subsection{Sampling error of gradient measurement \label{subsec:Sampling-error-of}}

Here, we analyze the statistical fluctuation under the limited runtime
of the trapped-ion quantum simulator.

Consider the target observable $O$ and the output state $\ket{\psi_{\theta}}=U_{T}...U_{i}(\theta)...U_{1}\ket{\psi_{0}}$,
where $\theta$ controls one of the single-qubit rotation gates $U_{i}(\theta)=e^{-i\theta G}$.
The gradient of $\theta$ can be determined by the parameter shift
method \citep{schuldEvaluatingAnalyticGradients2019}, in which the
same observable is evaluated at two shifted values of the parameters

\[
\partial_{\theta}\langle\psi_{\theta}|O|\psi_{\theta}\rangle=\langle\psi_{\theta+\pi/2}|O|\psi_{\theta+\pi/2}\rangle-\langle\psi_{\theta-\pi/2}|O|\psi_{\theta-\pi/2}\rangle.
\]

In principle, we can evaluate the gradient to any desired accuracy
by an arbitrary number of samples. In practice, however, the number
of samples we can obtain is limited by the experimental time cost
and the stability of the quantum device. This is more severe for the
trapped-ion quantum computers than the superconducting ones because,
despite the same scaling, the elementary gate operations for trapped-ion
are typically slower by a factor of about $10^{3}$.

To simulate statistical fluctuations, we align with typical experimental
parameters as described in Appendix \ref{sec:Sampling-Error-of}.
We estimate that the most costly molecule $\mathrm{F_{2}}$ would
take about $10$ hours, which is a feasible time cost for current
devices.

\subsection{Circuit drift error}

Another dominant source of error, circuit drift error, arises from
the miscalibration of experimental parameters or their drift during
the experiment. For instance, the drift in trap frequency can affect
the global interaction Hamiltonian in Eq. \ref{eq:H_TFIM}, while
the drift in laser intensity can alter both the single-qubit rotation
angle and the global interaction strength. For sufficiently slow parameter
drifts, the gradient descent process of VQE can dynamically adapt
to compensate for their effects. Therefore, we mainly consider the
shot-to-shot parameter fluctuation when sampling individual observables.
Specifically, we model the global interaction strength $J_{0}$ and
the parameters $\left\{ t_{d},\theta_{d}^{i}\right\} $ as being independently
subjected to the same level of fluctuation $x\rightarrow x(1+\mathcal{N}(0,\epsilon))$
where $\mathcal{N}(0,\epsilon)$ represents normal distribution with
$\epsilon=0.01$. 

\subsection{Simulation results \label{subsec:Error-analysis}}

Under the error model discussed above, we simulate the ground state
energy curves under realistic experimental parameters, as shown in
Fig. \ref{fig:with-noise}. The optimization converges on different
local minima due to random initial parameters, so we average the obtained
energy across multiple optimization rounds to characterize typical
experimental performance. In practice, better results can be achieved
if we optimize for multiple rounds and use the minimum, rather than
the average, of these values as the output.

For $\mathrm{H_{2}}$ and $\mathrm{LiH}$, the simulation result closely
matches the exact ground state of the encoded observable. The circuit
drift error dominates for $\mathrm{H_{2}}$, and the sampling error
dominates for $\mathrm{LiH}$. The precision of $\mathrm{H_{2}}$
achieves chemical accuracy for some bond distances, where the output
energy generated by noisy parameters has an error lower than the noise
itself, demonstrating the robustness of HEA-TI against noise. $\mathrm{F_{2}}$
suffers highly from insufficient circuit depth, resulting in large
systematic error regardless of the noise level. 

Optimizing variational parameters becomes more challenging as the
number of circuit parameters and the system size grows. As the system
size increases, the error caused by circuit drift scales linearly.
In comparison, the sampling error grows more rapidly. This is attributed
to the exponentially vanishing gradient, known as the barren plateau
phenomenon \citep{mccleanBarrenPlateausQuantum2018,wangNoiseinducedBarrenPlateaus2021},
which makes sampling increasingly inaccurate under fixed sampling
time. Consequently, for complex molecules, both error sources can
be dominant and must be suppressed to achieve chemical accuracy. The
circuit drift error can be mitigated by improving system stability,
while the sampling error can be addressed by employing more suitable
optimization algorithms. Also, various error mitigation methods and
barren plateau mitigation strategies can be employed to achieve chemical
accuracy, as will be discussed in Sec. \ref{sec:Discussions}. 

\section{Discussion \label{sec:Discussions}}

In this work, we study HEA-TI, a hardware-efficient ansatz tailored
for trapped-ion quantum simulators. By leveraging single-qubit operations
and global Hamiltonian evolution under native ion-ion interactions
in trapped-ion systems, this ansatz eliminates the need for individually
addressed multi-qubit gates, whose large-scale implementation poses
challenges on near-term quantum devices. We apply this ansatz to two
distinct problems in quantum state engineering and in quantum chemistry,
and we demonstrate its efficiency in both scenarios. We also analyze
the experimental error when implementing this ansatz on near-term
devices. When compared to the conventional gate-based methods, HEA-TI
significantly cuts down on the required resources for trapped-ion
systems, making it preferable for near-term trapped-ion platforms.

As estimated in Sec. \ref{subsec:Error-analysis}, reasonable noise
levels in near-term devices can easily lead to an error above  chemical
accuracy in electronic-structure problems. This is common in current
experiments \citep{omalleyScalableQuantumSimulation2016,streifSolvingQuantumChemistry2019,guoExperimentalQuantumComputational2024,zhaoOrbitaloptimizedPaircorrelatedElectron2023}
and can be alleviated by various error mitigation methods, including
probabilistic error cancellation \citep{temmeErrorMitigationShortDepth2017,sunMitigatingRealisticNoise2021},
gate error learning by Clifford fitting \citep{czarnikErrorMitigationClifford2021},
symmetry verification \citep{bonet-monroigLowcostErrorMitigation2018},
and readout error mitigation \citep{bravyiMitigatingMeasurementErrors2021},
which are widely applied in NISQ experiments.

While HEA-TI offers notable improvements in experimental feasibility
in trapped-ion devices, it shares a common drawback with other hardware-efficient
ansatzes: a complicated optimization process for the large number
of parameters. As discussed above, HEA-TI has an advantage of naturally
avoiding excessive entanglement, which can hinder learning \citep{ortizmarreroEntanglementInducedBarrenPlateaus2021}.
Still, when the system size and the circuit depth increases, one would
ultimately face the well-known challenges such as local minima trapping
\citep{anschuetzQuantumVariationalAlgorithms2022} and the barren
plateau phenomenon \citep{mccleanBarrenPlateausQuantum2018,wangNoiseinducedBarrenPlateaus2021}.
Over the past few years, various methods have been developed to address
these problems, including initialization \citep{grantInitializationStrategyAddressing2019},
circuit ansatz \citep{liuMitigatingBarrenPlateaus2022,pesahAbsenceBarrenPlateaus2021,liuPresenceAbsenceBarren2022},
cost function \citep{cerezoCostFunctionDependent2021}, and optimizer
\citep{chenLocalMinimaQuantum2023}. 

\section{Acknowledgments}

We acknowledge helpful discussions with Xiaoming Zhang, Lin Yao, Wending
Zhao, Jinzhao Sun, and Yueyuan Chen. This work was supported by the
Innovation Program for Quantum Science and Technology (No.2021ZD0301601),
the Shanghai Qi Zhi Institute, the Tsinghua University Initiative
Scientific Research Program and the Ministry of Education of China
through its fund to the IIIS. Y.K.W. acknowledges support from the
Tsinghua University Dushi program from Tsinghua University.

\appendix

\section{Parameter Shift Method for $t_{i}$ \label{sec:Parameter-Shift-Method}}

In this section, we describe our extension of the parameter shift
method that allows obtaining the gradient of evolution time $t_{i}$
in the HEA-TI protocol.

We initialize the system as $|\psi_{0}\rangle$ and apply unitary
operations $U=U_{T}...U_{i}(\theta)...U_{1}$, where $U_{i}(\theta)=e^{-i\theta G}$
is parameterized by $\theta$. The expectation value function of $O$
is given by

\[
f(\theta)=\langle\psi_{0}|U^{\dagger}OU|\psi_{0}\rangle
\]

This can be simplified by defining $O_{i}=U_{i}^{\dagger}...U_{T}^{\dagger}OU_{T}...U_{i}$
and $|\psi_{i}\rangle=U_{i}...U_{1}|\psi_{0}\rangle$
\[
f(\theta)=\langle\psi_{i-1}|U_{i}^{\dagger}(\theta)O_{i+1}U_{i}(\theta)|\psi_{i-1}\rangle
\]

The gradient of $\theta$, by definition, is
\begin{equation}
\frac{\partial f}{\partial\theta}=i\langle\psi_{i-1}|U_{i}^{\dagger}(\theta)\left[G,O_{i+1}\right]U_{i}(\theta)|\psi_{i-1}\rangle\label{eq:gradient_def}
\end{equation}
Our aim is to measure it directly from the quantum circuit.

We first review the commonly used parameter shift method \citep{schuldEvaluatingAnalyticGradients2019}
\begin{fact}[Parameter shift method]
When $G^{2}=I$, the gradient can be measured by adding an additional
gate $U_{\pm}=e^{\mp i\frac{\pi}{4}G}$ into the circuit. 
\end{fact}
\begin{proof}
We have $U_{i}(\theta)=\cos\theta-i\sin\theta G$. 
\begin{align*}
 & -i\left(U_{+}^{\dagger}O_{i+1}U_{+}-U_{-}^{\dagger}O_{i+1}U_{-}\right)\\
 & =\frac{i}{2}\left[(1-iG)O_{i+1}(1+iG)-(1+iG)O_{i+1}(1-iG)\right]\\
 & =GO_{i+1}-O_{i+1}G
\end{align*}

Substituting into Eq. \ref{eq:gradient_def}, we get 
\[
\frac{\partial f}{\partial\theta}=\langle\psi_{i-1}|U_{i}^{\dagger}(\theta)\left(U_{+}^{\dagger}O_{i+1}U_{+}-U_{-}^{\dagger}O_{i+1}U_{-}\right)U_{i}(\theta)|\psi_{i-1}\rangle
\]
\end{proof}
Note that the above also applies to two-qubit gates, e.g. $G=X\otimes X$.

We then provide an extension of the parameter shift method. The following
can be applied directly in the HEA-TI protocol to obtain the gradient
of evolution time $t_{i}$.
\begin{fact}
When $G=\sum_{P,Q,i,j}a_{i,j}^{P,Q}P_{i}Q_{j}$ where $P,Q$ are Pauli
operators and $a_{i,j}^{P,Q}$ are coefficients, the gradient can
still be measured by adding $U_{\pm}^{i,j}=e^{\mp i\frac{\pi}{4}P_{i}Q_{j}}$. 
\end{fact}
\begin{proof}
We have 
\[
\left[G,O_{i+1}\right]=\sum a_{i,j}^{P,Q}\left[P_{i}Q_{j},O_{i+1}\right]
\]

So the gradient of $\theta$ can be decomposed into the sum of measured
gradients
\begin{align*}
\frac{\partial f}{\partial\theta}=\sum a_{i,j}^{P,Q} & \langle\psi_{i-1}|U_{i}^{\dagger}(\theta)\\
 & \left(\left.U_{+}^{i,j}\right.^{\dagger}O_{i+1}U_{+}^{i,j}-\left.U_{-}^{i,j}\right.^{\dagger}O_{i+1}U_{-}^{i,j}\right)\\
 & U_{i}(\theta)|\psi_{i-1}\rangle
\end{align*}

To obtain the gradient, we sample from all possible $i,j,P,Q$ with
probability proportional to $\left|a_{i,j}^{P,Q}\right|$.
\end{proof}

\section{Encoding of Quantum Chemistry Problems \label{sec:Problem-Encoding-and}}

In this section, we introduce the problem encoding method used in
the main text for the quantum chemistry problem \citep{peruzzoVariationalEigenvalueSolver2014}. 

For an electronic molecular Hamiltonian, we apply the Born-Oppenheimer
approximation that takes the nuclear coordinates as parameters of
the electronic structure. After the second-quantization, we represent
the molecular Hamiltonian in a discrete basis set. We used the PySCF
\citep{sunSCFPythonBased2018} software suite to compute the molecular
integrals necessary to define the second-quantized Hamiltonian. 

We perform calculations for the most important orbitals, known as
the reduced active space. Specifically, we reduce the spin orbitals
with expected occupation number of electrons close to 0 or 1 \citep{mcardleQuantumComputationalChemistry2020},
where the occupation numbers are determined by classically tractable
methods. See \citep{guoExperimentalQuantumComputational2024} for
details of the selected orbitals in the selected active space for
the three molecules where their irreducible representations are listed.

The molecular Hamiltonian after orbital reduction can be expressed
as
\[
\hat{H}=\sum_{i,j}h_{ij}\hat{a}_{i}^{\dagger}\hat{a}_{j}+\frac{1}{2}\sum_{i,j,k,l}g_{ijkl}\hat{a}_{i}^{\dagger}\hat{a}_{j}^{\dagger}\hat{a}_{k}\hat{a}_{l},
\]
where $\hat{a}_{i}$ and $\hat{a}_{i}^{\dagger}$ denote the fermionic
annihilation and creation operators associated with the $i$-th spin-orbital
in the active space, respectively. The coefficients $h_{ij}$ and
$g_{ijkl}$ are the one- and two-electron integrals that can be evaluated
classically.

Using the Jordan-Wignar transformation \citep{jordanUberPaulischeAquivalenzverbot1928},
we map the molecular Hamiltonian to the qubit system for evaluation
on a quantum processor. Specifically, each fermionic operator $\hat{a}_{j}$
is mapped as $\frac{1}{2}\left(X_{j}+iY_{j}\right)\bigotimes_{k=1}^{j-1}Z_{k}$,
where $X_{j},Y_{j},Z_{j}$ are Pauli operators acting on the $j$-th
qubit.

The initial state, also referred to as the reference state, is selected
by a classical approximation of the molecule. The molecules considered
in this work are closed-shell systems, with all orbitals doubly occupied.
Therefore, we adopt the restricted Hartree-Fock method. The Hartree-Fock
state is constructed so that half of the electrons occupy the lowest
energy spin orbitals in the spin $\alpha$ sector, and the remaining
half occupy the spin $\beta$ sector in the same way. 

As an example, for $\mathrm{LiH}$, we freeze the lowest two molecular
orbits and select three molecular orbits as the active space. After
reduction, the system contains two electrons and three molecular orbits.
The Hartree-Fock state is $|100100\rangle$, where the first and last
$3$-qubit state $|100\rangle$ correspond to the spin $\alpha$ and
$\beta$ parts, respectively. 

\section{Sampling Error of Gradient \label{sec:Sampling-Error-of}}

Under the settings described in Sec. \ref{subsec:Sampling-error-of},
we obtain the gradient of a parameter $\theta$ by sampling 
\[
\langle\psi_{\theta+\pi/2}|O|\psi_{\theta+\pi/2}\rangle-\langle\psi_{\theta-\pi/2}|O|\psi_{\theta-\pi/2}\rangle
\]

The target observable can be decomposed into the Pauli basis $O=\sum c_{i}P_{i}$,
where $P_{i}$ are Pauli operators. Its expectation value is sampled
by measuring $\ket{\psi_{\theta}}$ in Pauli basis $P_{i}$. The sampling
error of this step can be simulated as $\frac{\sigma}{\sqrt{M}}$,
where $\sigma$ is the standard deviation of measurement outcome and
$M$ is the number of samples we can perform.

To estimate the standard deviation, we model the measurement outcomes
of each term $P_{i}$ for $|\psi_{\theta\pm\pi/2}\rangle$ as independent
and identically distributed random variables $X_{i,\pm}$. Given a
limited number of $2M$ samples in each step, the best choice is to
allocate $M^{i}$ according to the variance for each $c_{i}X_{i,\pm}$.
Since we have no prior knowledge of $\left\langle X_{i}\right\rangle $,
we allocate $M^{i}=\frac{|c_{i}|}{\sum|c_{i}|}M$ samples for $X_{i,\pm}$,
respectively. The standard deviation of $\partial_{\theta}\langle\psi_{\theta}|O|\psi_{\theta}\rangle$
is then calculated by $\sigma^{2}=\sum_{i}\frac{c_{i}^{2}}{M^{i}}(\sigma_{i,+}^{2}+\sigma_{i,-}^{2})$,
where $\sigma_{i,\pm}^{2}=1-\langle X_{i,\pm}\rangle^{2}$. Therefore,
we can simulate the obtained gradient value following a normal distribution
$\mathcal{N}(\mu,\sigma^{2})$ where $\sigma\propto\frac{1}{\sqrt{M}}$.

Below, we analyze the time cost for the $\mathrm{F_{2}}$ molecule
which requires the most quantum resource. We use the Adam optimizer
\citep{kingmaAdamMethodStochastic2017} and limit the gradient descent
process to no more than $120$ steps. In each step, we update all
the $n_{\text{p}}=ND$ parameters opposite to the direction of the
gradient. We adopt a typical experimental parameter $J_{0}=2\pi\times0.5\text{ kHz}$
\citep{zhangObservationManybodyDynamical2017,liProbingCriticalBehavior2023}.
The Hamiltonian evolution time for $\mathrm{F_{2}}$ with $t_{\text{tot}}=Dt_{0}=3\times0.4=1.2$
would thus correspond to a physical time of $\frac{1.2}{2\pi\times500}\approx0.4\text{ ms}$.
The time cost for single-qubit gates is significantly smaller. Considering
laser cooling, pulse sequence, and CCD data collection, we estimate
the total time cost as $4.5\text{ ms}$ for each shot. We adjust the
measurement precision for the first few gradient steps by limiting
the number of shots at the $k$-th step as $2M_{k}=\frac{k^{2}}{\left(1+k\right)^{2}}2M_{0}$,
where $M_{0}$ is an overall parameter. We set $M_{0}=1000$.

Using the overlapped grouping method \citep{wuOverlappedGroupingMeasurement2023,valluryQuantumComputedMoments2020},
samples of more than one operators can be obtained from each Pauli-basis
measurement. Suppose the target observable can be decomposed into
Pauli operators as $O=\sum c_{i}P_{i}$. When two observables $P_{i}$
and $P_{j}$ are qubit-wise compatible, we can measure them simultaneously
on a single measurement basis. We estimate that the variance of the
observable estimator is reduced by half, which is typical for numerical
simulation \citep{wuOverlappedGroupingMeasurement2023}. 

From the analysis above, the total sampling time for $\mathrm{F_{2}}$
can be estimated by 
\[
n_{\text{p}}\times\sum_{k=1}^{120}2M_{k}\times4.5\text{ ms}=10\text{ hour}
\]
which is a reasonable cost for current devices. The time costs for
$\mathrm{H_{2}}$ and $\mathrm{LiH}$ are shorter due to less number
of gradient steps, less number of parameters, and a stronger interaction
strength $J_{0}$ achievable in ion chains with fewer qubits. We estimate
them to be less than $2$ hours and $5$ hours, respectively.

\bibliography{MyLibrary}

\end{document}